\documentclass[11pt,letterpaper]{amsart}

\usepackage[left=1in,top=1in,right=1in,bottom=1in]{geometry}
\usepackage[foot]{amsaddr}
\usepackage{amsmath}
\usepackage{amsthm}
\usepackage{amssymb}
\usepackage{url}
\usepackage{hyperref}
\usepackage{todonotes}
\usepackage{mathtools}
\usepackage{mathrsfs}
\usepackage[T1]{fontenc}
\usepackage{accents}
\usepackage{comment}
\pdfoutput=1
 
\newtheorem{Thm}{Theorem}
\newtheorem{Lem}[Thm]{Lemma}
\newtheorem{Rem}[Thm]{Remark}
\newtheorem{Prop}[Thm]{Proposition}

\newtheorem{Cor}[Thm]{Corollary}
\newtheorem{Conj}[Thm]{Conjecture}
\theoremstyle{remark}
\newtheorem{Exa}[Thm]{Example}
\theoremstyle{definition}
\newtheorem{Def}[Thm]{Definition}

\newcommand{\N}{\ensuremath{\mathbb{N}}}
\newcommand{\Z}{\ensuremath{\mathbb{Z}}}
\newcommand{\R}{\ensuremath{\mathbb{R}}}

\newcommand{\F}{\ensuremath{\mathbb{F}}}
\newcommand{\C}{\ensuremath{\mathbb{C}}}
\newcommand{\sF}{\ensuremath{\mathscr{F}}}
\newcommand{\sT}{\ensuremath{\mathscr{T}}}
\newcommand{\cT}{\ensuremath{\mathcal{T}}}
\newcommand{\cU}{\ensuremath{\mathcal{U}}}
\newcommand{\cD}{\ensuremath{\mathcal{D}}}
\newcommand{\PP}{\ensuremath{\mathbb{P}}}

\newcommand{\sC}{\ensuremath{\mathscr{C}}}

\newcommand{\MM}{\ensuremath{\mathrm{MM}}}
\newcommand{\Kp}{\ensuremath{\mathrm{K}}}
\newcommand{\Ra}{\ensuremath{\mathrm{R}}}
\newcommand{\ARa}{\ensuremath{\mathrm{\tilde R}}}
\newcommand{\iv}[1]{[\![{#1}]\!]}

\def\ot{\otimes}
\def\SS{\mathfrak{S}}

\begin{document}


\title[]{A universal sequence of tensors\\{} for the asymptotic rank conjecture}

\author{Petteri Kaski}
\address{Aalto University, Department of Computer Science}
\email{petteri.kaski@aalto.fi}

\author{Mateusz Micha\l{}ek}
\address{Universit\"at Konstanz, Fachbereich Mathematik und Statistik}
\email{mateusz.michalek@uni-konstanz.de}

\begin{abstract}
The \emph{exponent} $\sigma(T)$ of a tensor $T\in\mathbb{F}^d\otimes\mathbb{F}^d\otimes\mathbb{F}^d$ over a field $\mathbb{F}$ captures the base of the exponential growth rate of the tensor rank of $T$ under Kronecker powers. Tensor exponents are fundamental from the standpoint of algorithms and computational complexity theory; for example, the exponent $\omega$ of square matrix multiplication can be characterized as $\omega=2\sigma(\mathrm{MM}_2)$, where $\mathrm{MM}_2\in\mathbb{F}^4\otimes\mathbb{F}^4\otimes\mathbb{F}^4$ is the tensor that represents $2\times 2$ matrix multiplication.

Strassen [FOCS 1986] initiated a duality theory for spaces of tensors that enables one to characterize the exponent of a tensor via objects in a dual space, called the \emph{asymptotic spectrum} of the primal (tensor) space. While Strassen's theory has considerable generality beyond the setting of tensors---Wigderson and Zuiddam [\emph{Asymptotic Spectra:~Theory, Applications, and Extensions}, preprint, 2023] give a recent exposition---progress in characterizing the dual space in the tensor setting has been slow, with the first universal points in the dual identified by Christandl, Vrana, and Zuiddam~[\emph{J.~Amer.~Math.~Soc.}~36~(2023)]. In parallel to Strassen's theory, the algebraic geometry community has developed a geometric theory of tensors aimed at characterizing the structure of the primal space and tensor exponents therein; the latter study was motivated in particular by an observation of Strassen (implicit in [\emph{J.~Reine Angew.~Math.}~384~(1988)]) that matrix-multiplication tensors have limited universality in the sense that $\sigma(\mathbb{F}^d\otimes\mathbb{F}^d\otimes\mathbb{F}^d)\leq \frac{2\omega}{3}=\frac{4}{3}\sigma(\mathrm{MM}_2)$ holds for all $d\geq 1$. In particular, this limited universality of the tensor $\mathrm{MM}_2$ puts forth the question whether one could construct explicit universal tensors that exactly characterize the worst-case tensor exponent in the primal space. Such explicit universal objects would, among others, give means towards a proof or a disproof of Strassen's asymptotic rank conjecture [\emph{Progr.~Math.}~120 (1994)]; the former would immediately imply $\omega=2$ and, among others, refute the Set Cover Conjecture (cf.~Bj\"orklund and Kaski [STOC~2024] and Pratt [STOC~2024]). 

Our main result is an explicit construction of a \emph{sequence} $\mathcal{U}_d$ of zero-one-valued tensors that is universal for the worst-case tensor exponent; more precisely, we show that $\sigma(\mathcal{U}_d)=\sigma(d)$ where $\sigma(d)=\sup_{T\in\mathbb{F}^d\otimes\mathbb{F}^d\otimes\mathbb{F}^d}\sigma(T)$. We also supply an explicit universal sequence $\mathcal{U}_\Delta$ localised to capture the worst-case exponent $\sigma(\Delta)$ of tensors with support contained in $\Delta\subseteq [d]\times[d]\times [d]$; by combining such sequences, we obtain a universal sequence $\mathcal{T}_d$ such that $\sigma(\mathcal{T}_d)=1$ holds if and only if Strassen's asymptotic rank conjecture holds for $d$. Finally, we show that the limit $\lim_{d\rightarrow\infty}\sigma(d)$ exists and can be captured as $\lim_{d\rightarrow\infty} \sigma(D_d)$ for an explicit sequence $(D_d)_{d=1}^\infty$ of tensors obtained by diagonalisation of the sequences $\mathcal{U}_d$.

As our second result we relate the \emph{absence} of polynomials of fixed degree vanishing on tensors of low rank, or more generally asymptotic rank, with upper bounds on the exponent $\sigma(d)$. Using this technique, one may bound asymptotic rank for all tensors of a given format, knowing enough specific tensors of low asymptotic rank.
\end{abstract}

\maketitle


\section{Introduction}

\label{sect:introduction}

\subsection{Exponents of tensors and the quest for universality}

For an infinite field $\F$ and a positive integer constant~$d$, 
the {\em exponent} 
of a nonzero tensor $T\in\F^d\otimes\F^d\otimes\F^d$ is the least nonnegative 
real~$\sigma(T)$ such that the sequence of tensor (Kronecker) powers
\begin{equation}
\label{eq:box-seq}
T^{\boxtimes\mathbb{N}}=(T^{\boxtimes p}:p=1,2,\ldots)
\end{equation}
has its sequence of tensor ranks bounded 
by $\Ra(T^{\boxtimes p})\leq d^{\sigma(T)p+o(p)}$.
Exponents of specific {\em constant-size} tensors $T$ are of fundamental 
significance in the study of algorithms and algebraic complexity 
theory~\cite{burgisser2013algebraic, WigdersonZ2023}. For example, 
Strassen showed that the exponent $\omega$ of 
square matrix multiplication satisfies 
$\omega=2\sigma(\MM_2)$~\cite{Strassen1986,Strassen1988}, 
where $\MM_2$ is the $4\times 4\times 4$ tensor representing the 
multiplication map for two $2\times 2$ matrices. Recently, it was shown 
that the Set Cover 
Conjecture~\cite{CyganDLMNOPSW2016,CyganFKLMPPS15,KrauthgamerT2019}
fails if the exponent 
of a specific $7\times 7\times 7$ tensor $Q$ is sufficiently close 
to one~\cite{BjorklundK2023,Pratt2023}.

The study of exponents of tensors---or, what is the same up to exponentiation, 
the study of {\em asymptotic rank}\footnote{The {\em asymptotic rank} 
of $T\in\F^d\otimes\F^d\otimes\F^d$ is 
$\ARa(T)=\lim_{p\rightarrow\infty}\Ra(T^{\boxtimes p})^{1/p}$ and we have $\ARa(T)=d^{\sigma(T)}$.}~\cite{Gartenberg1985}
of tensors---is difficult. 
The exponent $\omega$ of square matrix multiplication is perhaps the best 
studied nontrivial exponent, and even in its case the best current lower 
bound $\omega\geq 2$ remains the trivial one 
(but see \cite{Blaser1999, landsberg2014new, landsberg2015new, landsberg2017geometry, Conner_Harper_Landsberg_2023}), 
and the best current upper bound 
$\omega\leq 2.371866$~\cite{DuanWZ2023} is a result of extensive work spanning 
decades (e.g.~~\cite{Strassen1969,Pan1978,BiniCRL1979,Schonhage1981,Romani1982,CoppersmithW1982,Strassen1986,CoppersmithW1990,Stothers2010,VassilevskaWilliams2012,LeGall2014,AlmanW2021}) and relying on increasingly sophisticated techniques. 
In parallel to the study of exponents of individual tensors, research 
effort has been invested into developing a structural theory for spaces
of tensors and their exponents, a theory to which the present paper also 
contributes. Two rough lines of research most relevant to our present work 
are as follows.

\medskip
\noindent
{\em Strassen's duality theory---the asymptotic spectrum of tensors.}
The first line of research, announced by Strassen~\cite{Strassen1986} in his
1986 FOCS paper and developed in a sequence of 
papers~\cite{Strassen1987,Strassen1988,Strassen1991} and 
PhD theses~\cite{Burgisser1990,Tobler1991,Mauch1998} (cf.~\cite{Strassen2005}),
builds a duality theory for asymptotic rank of tensors based on 
the theory of preordered commutative semirings, with the direct sum and 
tensor (Kronecker) product of tensors as the pertinent semiring operations, and 
the preorder defined by a rank-capturing tensor restriction relation. 
Wigderson and Zuiddam~\cite{WigdersonZ2023} give 
a recent comprehensive exposition of Strassen's theory and the 
preorder-theoretic topological dual spaces---the {\em asymptotic spectrum} 
of a space of tensors---which enable a tight characterization of the 
asymptotic rank of a tensor via the preorder-monotone homomorphisms 
in the dual. Yet, progress in terms of understanding the structure and
identifying explicit points in the asymptotic spectrum of tensors has been 
slow. Beyond Strassen's original construction of support functionals, 
which are restricted to the subspace of oblique tensors, only recently 
Christandl, Vrana, and Zuiddam~\cite{ChristandlVZ2023} constructed 
a family of explicit {\em universal} spectral points---called quantum 
functionals---using theory of quantum entropy and covariants; 
however, also this dual family is far from yielding broadly tight lower 
bounds on asymptotic rank.

Viewed from the standpoint of Strassen's duality theory, rather than 
working in a dual space, in this paper we seek and obtain as our main result 
an explicit and universal {\em primal} characterization of tensor exponents 
and thus asymptotic rank in $\F^d\ot\F^d\ot\F^d$---before proceeding to our 
results, let us review a second pertinent line of prior research and motivation.

\medskip
\noindent
{\em Geometry of tensors and the asymptotic rank conjecture.}
The second line of research seeks to study spaces of tensors to identify the 
{\em worst-case} behaviour in the space using tools from algebraic geometry,
in particular building on the seminal concept of {\em border rank} of a tensor 
due to Bini, Capovani, Romani, and Lotti~\cite{BiniCRL1979} 
(see also Sch\"onhage~\cite{Schonhage1981}) and its 
geometric characterization via secant varieties of  Segre and Veronese varieties
(e.g.~\cite{buczynska2021apolarity, buczynski2013ranks, landsberg2010ranks, Landsberg2012, Landsberg2019, zak1993tangents}).
For a space $\sF$ of tensors, let $\sigma(\sF)=\sup_{T\in\sF}\sigma(T)$;
as a special case, for $d=1,2,\ldots$
let $\sigma(d)=\sigma(\F_d\otimes\F_d\otimes\F_d)$. 
It is immediate that $1\leq\sigma(d)\leq 2$ and
that $\sigma(1)=1$. Over the complex numbers, it is a nontrivial consequence 
of the geometry of tensors that $\sigma(2)=1$.
Already the value $\sigma(3)$ 
is open and would be of substantial interest to determine. 
For example, it is known that
$\sigma(3)=1$ implies $\omega=2$ by the application of the Coppersmith--Winograd method \cite{CoppersmithW1990} to a specific $3\times 3\times 3$ tensor.
It appears very difficult to prove lower bounds on $\sigma(d)$ apart
from the trivial $\sigma(d)\geq 1$. Indeed, any tensor $T\in (\F^d)^{\otimes 3}$ with $\sigma(T)>1$ would yield an explicit sequence of tensors, namely its Kronecker powers, such that for any constant $C>0$ for $k\gg 0$ we have rank and border rank of $T^{\boxtimes k}$ greater than $C d^k$. Currently, we do not know explicit examples of such sequences for $C=3$ with the current world record lower bound on rank in  \cite{alexeev2011tensor}. One of the main obstacles is lack of methods to prove that a given tensor has high rank or border rank. In case of rank, the state of the art is the substitution method, based on linear algebra. In case of border rank, one looks for polynomial witnesses that vanish on tensors of given bounded border rank. However in this case, most known equations, not found via explicit computations, also vanish on so-called cactus varieties, which fill the ambient space quickly and thus known methods cannot provide super-linear lower bounds on tensor rank \cite{bernardi2013cactus, bernardi2020waring,  buczynska2014secant, galkazka2017vector, iarrobino1999power}. The state of the art for border rank is based on mixture of different methods, barely breaking the bound for $C=2$ \cite{landsberg2019towards}. As the rank of the generic tensor grows quadratically with $d$ the problem is often referred to as an instance of the \emph{hay in a haystack} problem
(phrase due to H. Karloff): find an explicit object that behaves generically. There is hope in the new introduced method of border apolarity \cite{buczynska2021apolarity}, still there is currently no clear path of getting past $C=3$. 

The difficulty of lower bounds and progressively improving upper bounds for 
specific exponents, the exponent $\omega$ in particular, has prompted bold 
conjectures on 
worst-case exponents for broad families of tensors. Most notably, 
writing $\sT_d$ for the space of all tight
tensors in $\F^d\otimes\F^d\otimes\F^d$, {\em Strassen's asymptotic rank 
conjecture} (cf.~\cite[Conjecture~5.3]{Strassen1994}) states that the 
worst-case exponent for this space is the least possible:

\begin{Conj}[Strassen's asymptotic rank conjecture]
\label{conj:Strassen}
For all $d\geq 1$ it holds that $\sigma(\sT_d)=1$.
\end{Conj}

Strassen's asymptotic rank conjecture, if true, immediately implies the
algorithmically serendipitous corollary $\omega=2$ in 
particular; cf.~also~\cite{BjorklundK2023,Pratt2023}
for further consequences to algorithms.
A yet stronger conjecture (cf.~B\"urgisser, Clausen, 
and Shokrollahi~\cite[Problem~15.5]{burgisser2013algebraic}; 
also e.g.~Conner, Gesmundo, Landsberg, Ventura, and 
Wang~\cite[Conjecture~1.4]{ConnerGLVW2021} 
as well as Wigderson and Zuiddam~\cite[Section 13, p.~122]{WigdersonZ2023}) 
states that the least possible exponent is shared by 
all concise tensors in $\F^d\ot\F^d\ot\F^d$:

\begin{Conj}[Extended asymptotic rank conjecture]
\label{conj:Strassen-extended} 
For all $d\geq 1$ it holds that $\sigma(d)=1$.
\end{Conj}

A key result supporting Conjecture~\ref{conj:Strassen} and
Conjecture~\ref{conj:Strassen-extended}, implicit in 
Strassen~\cite[Proposition~3.6]{Strassen1988} and highlighted by 
Christandl, Vrana, and Zuiddam~\cite[Proposition~2.12]{ChristandlVZ2021} 
as well as Conner, Gesmundo, Landsberg, and 
Ventura~\cite[Remark~2.1]{ConnerGLV2022}, 
is that tensor rank is known to be nontrivially submultiplicative
under Kronecker products; stated in terms of worst-case tensor exponents, 
Strassen proved that for all $d\geq 1$ it holds that $\sigma(d)\leq \frac{2\omega}{3}$.

Viewed in terms of exponents of tensors and universality, we can rephrase 
Strassen's result as stating that the exponent of the matrix multiplication 
tensor $\MM_2$ controls from above the exponent of all other tensors, namely 
we have $\sigma(d)\leq \frac{4}{3}\sigma(\MM_2)=\frac{2\omega}{3}$. 
This rephrasing
suggests that one should seek explicit constructions of tensors 
$U\in\F^d\otimes\F^d\otimes\F^d$ that are worst-case {\em universal}
for the class with $\sigma(d)=\sigma(U)$. In rough analogy with 
computational complexity theory, such explicit tensors capture 
the ``hardest'' tensors in a class of tensors and, for example, would 
provide an explicit object of study towards resolving Conjectures
\ref{conj:Strassen} and~\ref{conj:Strassen-extended}.

\subsection{Our results---explicit universal sequences of zero-one-valued tensors}

While we do not present explicit individual tensors $U$ that are 
universal, as our main result we present an explicit {\em sequence} $\cU$ 
of tensors that is universal and consists of zero-one-valued tensors 
in coordinates. Towards this end, let us extend the definition 
of the exponent of a tensor $T$ to a sequence
\begin{equation}
\label{eq:tensor-seq}
\cT=(T_j\in\F^{s_j}\otimes\F^{s_j}\otimes\F^{s_j}:j=1,2,\ldots)
\end{equation}
of nonzero tensors. The {\em exponent} of the sequence $\cT$ is the 
least nonnegative real $\sigma(\cT)$ such that 
$\Ra(T_j)\leq s_j^{\sigma(\cT)+o_j(1)}$. 
From \eqref{eq:box-seq} and \eqref{eq:tensor-seq} we immediately have 
$\sigma(T^{\boxtimes\mathbb{N}})=\sigma(T)$ for all tensors 
$T\in\F^d\otimes\F^d\otimes\F^d$. 
Our main result is that there is an explicit sequence of tensors that characterizes the exponent of $\F^d\otimes\F^d\otimes\F^d$
and thus enables an approach towards resolving Strassen's conjecture.

\begin{Thm}[Main; A universal sequence of tensors for $d$ fixed]
\label{thm:main}
For all $d\geq 1$ there is an explicit sequence $\cU_d$ 
of zero-one-valued tensors with $\sigma(\cU_d)=\sigma(d)$. 
\end{Thm}

In coordinates, the $q$\textsuperscript{th} 
tensor in the sequence $\cU_d$ admits explicit 
combinatorial expression as a union of orbit-indicator tensors under 
a particular action of the symmetric group $\SS_q$. This combinatorial 
structure enables us to study the exponent $\sigma(\Delta)$
of the space of tensors spanned by 
$\{e_i\otimes e_j\otimes e_k:(i,j,k)\in\Delta\}\subseteq\F^d\otimes\F^d\otimes\F^d$ for a nonempty $\Delta\subseteq[d]\times[d]\times[d]$; when equality
holds, we have $\sigma(\Delta)=\sigma(d)$. 

\begin{Thm}[Support-localized universal sequences of tensors]
\label{thm:main-localized}
For all nonempty $\Delta\subseteq[d]\times[d]\times[d]$ there is 
an explicit sequence $\cU_\Delta$ of zero-one-valued tensors 
with $\sigma(\cU_\Delta)=\sigma(\Delta)$. 
\end{Thm}

From Theorem~\ref{thm:main-localized} we obtain the following corollary 
for tight tensors and Strassen's conjecture (Conjecture~\ref{conj:Strassen}).
We need short preliminaries. 
We say that the set $\Delta$ is {\em tight} if there exist injective functions
$\alpha, \beta,\gamma:[d]\rightarrow\Z$
such that $\alpha(i)+\beta(j)+\gamma(k)=0$ for all $(i,j,k)\in\Delta$. 
A tensor $T\in\F^d\ot\F^d\ot\F^d$ is {\em tight} if it admits a 
coordinate-representation whose support is contained in a tight set.

\begin{Thm}[A universal sequence of tensors for Strassen's conjecture]
\label{thm:tight-universal}
For all $d\geq 1$ there is an explicit sequence $\cT_d$ of tight 
zero-one-valued tensors with $\sigma(\cT_d)=\sigma(\sT_d)$.
\end{Thm}

Further, by diagonalising the sequences $\cU_d$ for 
increasing $d$ we obtain a universal sequence providing worst possible 
exponent irrespective of $d$; that is, a universal sequence
for the extended asymptotic rank conjecture 
(Conjecture~\ref{conj:Strassen-extended}):

\begin{Thm}[A universal sequence of tensors for the extended asymptotic rank conjecture] 
\label{thm:extended-universal}
There is an explicit sequence $\cD=(D_d:d=1,2,\ldots)$ of zero-one-valued tensors with $\lim_{d\rightarrow\infty} \sigma(D_d)=\lim_{d\rightarrow\infty}\sigma(d)$.
\end{Thm}
We note that the limit in the theorem above exists and is the supremum of 
all $\sigma(d)$; cf.~Lemma~\ref{lem:forbigthenforall}. The above theorem may be regarded as a solution to the aforementioned hay in the haystack problem for the exponent of tensors. 

\subsection{Overview of techniques}

Before proceeding to review our further results, it will be convenient to 
give an overview of our main techniques and concepts 
underlying Theorem~\ref{thm:main}. In particular, a basis induced from 
integer compositions for the linear span of the image of the 
Kronecker power map $S\mapsto S^{\boxtimes q}$ will be our key tool. 

\medskip
\noindent
{\em 1. The Kronecker power map\/ $\Kp_{d,q}$ in coordinates.}
For a field $\F$ and positive integers $d$ and $q$, our key object of study
is the {\em Kronecker power map}
\[
\Kp_{d,q}:\F^d\ot\F^d\ot\F^d\rightarrow \F^{d^q}\ot\F^{d^q}\ot\F^{d^q}
\]
that takes a tensor $S\in \F^d\ot\F^d\ot\F^d$ to its $q$\textsuperscript{th}
tensor (Kronecker) power $\Kp_{d,q}(S)=S^{\boxtimes q}$.

It will be convenient to study the Kronecker power map $\Kp_{d,q}$ by working 
with tensors in coordinates, so let us set up conventions accordingly. 
Let us write $[d]=\{1,2,\ldots,d\}$ and identify the spaces 
$\F^{d\times d\times d}=\F^d\otimes\F^d\otimes\F^d$.
For a tensor $S\in\F^{d\times d\times d}$ in the domain of $\Kp_{d,q}$,
we write $S_{i,j,k}\in\F$ for the entry of $S$ at position 
$i,j,k\in[d]$. 
For a tensor $T\in\F^{d^q\times d^q\times d^q}$ in the codomain of~$\Kp_{d,q}$, 
it is convenient to index the entries $T_{I,J,K}\in\F$ of $T$ with
$q$-tuples $I=(i_1,i_2,\ldots,i_q)\in [d]^q$,
$J=(j_1,j_2,\ldots,j_q)\in [d]^q$,
$K=(k_1,k_2,\ldots,k_q)\in [d]^q$. Indeed, with this convention,
the image $\Kp_{d,q}(S)=S^{\boxtimes q}$ of 
a tensor $S\in\F^{d\times d\times d}$ can be defined entrywise 
for all $I,J,K\in[d]^q$ by 
\begin{equation}
\label{eq:box-entry}
S^{\boxtimes q}_{I,J,K}=
S_{i_1,j_1,k_1}
S_{i_2,j_2,k_2}
\cdots
S_{i_q,j_q,k_q}\,.
\end{equation}

\medskip
\noindent
{\em 2. An explicit basis for the linear span of the image of\/ $\Kp_{d,q}$.}
Next we embed the image $\Kp_{d,q}(\F^{d\times d\times d})$
in the least-dimensional subspace of $\F^{d^q\times d^q\times d^q}$ 
possible, so-called linear span, and provide an explicit basis for this subspace in the
assumed coordinates. 
Towards this end, working with integer compositions that capture
the distinct right-hand sides of \eqref{eq:box-entry} for a generic
$S$ with support in a nonempty set $\Delta\subseteq[d]\times[d]\times[d]$ 
will be convenient. In precise terms, a function 
$g:\Delta\rightarrow\N_{\geq 0}$ with $\sum_{\delta\in\Delta}g(\delta)=q$ 
is a {\em composition} of the positive integer $q$ with domain~$\Delta$. 
Let us write $\sC^\Delta_q$ for the set of all compositions
of $q$ with domain~$\Delta$. 
The distinct right-hand sides 
of~\eqref{eq:box-entry} for a generic $S$ are now enumerated by the
compositions $g\in\sC^\Delta_q$; explicitly, associate 
with $g$ the product
\begin{equation}
\label{eq:s-g}
S^g=\prod_{i,j,k\in[d]}S_{i,j,k}^{g(i,j,k)}\in\F\,.
\end{equation}
Writing $\F^{\Delta}$ for the subspace of
all tensors in $\F^{d\times d\times d}$ with support in $\Delta$, we next
provide a basis 
$(T^{(g)}\in\F^{d^q\times d^q\times d^q}:g\in\sC^{\Delta}_q)$ 
that enables us to express 
the Kronecker power $S^{\boxtimes q}$ of an arbitrary tensor 
$S\in\F^{\Delta}$ as the linear combination
\begin{equation}
\label{eq:s-boxq-via-t-g}
S^{\boxtimes q}=\sum_{g\in\sC^{\Delta}_q}S^gT^{(g)}\,.
\end{equation}
We need short preliminaries to define the tensors $T^{(g)}$ in the assumed 
coordinates. For $I,J,K\in[d]^q$, define the triple-counting composition 
$\Phi_{I,J,K}\in\sC^{[d]\times [d]\times [d]}_q$ for all $u,v,w\in [d]$ by
\begin{equation}
\label{eq:phi}
\Phi_{I,J,K}(u,v,w)=|\{\ell\in [q]:i_\ell=u,\ j_\ell=v,\ k_\ell=w\}|\,.
\end{equation}
We use {\em Iverson's bracket notation}; for a logical proposition $P$, we write
$\iv{P}$ to indicate a $0$ if $P$ is false and a $1$ if $P$ is true.

\begin{Def}[Composition basis for the linear span of the image $\Kp_{d,q}(\F^\Delta)$]\label{def:Tg}
For $g\in \sC_q^{\Delta}$, define the tensor 
$T^{(g)}\in \F^{d^q\times d^q\times d^q}$ in coordinates 
for all $I,J,K\in [d]^q$ by
\begin{equation}
\label{eq:t-g}
\bigl(T^{(g)}\bigr)_{I,J,K}=\iv{\Phi_{I,J,K}=g}\,.
\end{equation}
The tensors $(T^{(g)}:g\in\sC_q^{\Delta})$ are
the {\em composition basis} for the linear span of $\Kp_{d,q}(\F^\Delta)$.
\end{Def}
We provide various equivalent characterisations of these tensors $T^{(g)}$ 
in Lemma~\ref{lem:kdp} (on page~\pageref{lem:kdp}). 
The characterisation via integer compositions 
in Definition~\ref{def:Tg} in particular
enables an immediate verification via \eqref{eq:t-g}, 
\eqref{eq:s-g}, and \eqref{eq:box-entry} that \eqref{eq:s-boxq-via-t-g} holds.
In particular, the linear span of the image $\Kp_{d,q}(\F^\Delta)$ is 
contained in the linear span of the composition basis. What is less immediate 
is that the reverse 
containment holds under mild assumptions on the field $\F$ by identifying 
the Kronecker power map with the Veronese map on homogeneous $\Delta$-variate 
polynomials of degree $q$. In coordinates, by relying on homogeneous polynomial
interpolation we show in Corollary \ref{cor:tgbasisofL} and Proposition \ref{prop:T^ginimage} 
that for every $g\in\sC_q^\Delta$ there exist
tensors $S_f\in\F^{\Delta}$ and scalars $\lambda_{f,g}\in\F$ indexed by 
$f\in\sC_q^\Delta$ such that
\begin{equation}
\label{eq:t-g-via-s-boxq}
T^{(g)}=\sum_{f\in\sC_q^\Delta}\lambda_{f,g}S_{\!f}^{\boxtimes q}\,.
\end{equation}
From \eqref{eq:s-boxq-via-t-g} and \eqref{eq:t-g-via-s-boxq} we have that the composition basis spans exactly the linear span of $\Kp_{d,q}(\F^\Delta)$.

\medskip
\noindent 
{\em 3. Sublinearity and serendipity of polynomial dimensionality of the span.}
By subadditivity of tensor rank, from the linear 
combination~\eqref{eq:s-boxq-via-t-g} we have immediately for an
arbitrary tensor $S\in\F^{\Delta}$ that 
\[
\Ra\bigl(S^{\boxtimes q}\bigr)\leq
\sum_{g\in\sC^\Delta_q}\Ra\bigl(T^{(g)}\bigr)\,.
\]
Conversely, for an arbitrary $g\in\sC^\Delta_q$, we have 
from~\eqref{eq:t-g-via-s-boxq} that
\[
\Ra\bigl(T^{(g)}\bigr)\leq
\sum_{f\in\sC^\Delta_q}\Ra\bigl(S_f^{\boxtimes q}\bigr)\,.
\]
Recalling from combinatorics of integer compositions that 
$|\sC^\Delta_q|=\binom{|\Delta|-1+q}{|\Delta|-1}\leq (|\Delta|-1+q)^{|\Delta|-1}$, we observe that dimension of the linear span of $\Kp_{d,q}(\F^\Delta)$ 
grows only polynomially in $q$ when $\Delta$ is fixed.
Theorem~\ref{thm:main} and Theorem~\ref{thm:main-localized} follow essentially
immediately by taking $\cU_\Delta=(U_{\Delta,q}:q=1,2,\ldots)$ with
$U_{\Delta,q}=\bigoplus_{g\in\sC_q^\Delta}T^{(g)}$ as the universal sequence.

\medskip
\noindent
{\em Structure and description of the tensors $T^{(g)}$.}
The tensors $T^{(g)}$ have several group-theoretic descriptions, e.g.~as orbit-indicators of the group $\SS_q$, see Lemma \ref{lem:kdp}. There are several important applications of group theory in the study or tensor rank and asymptotic rank. We note that when $G$ is an Abelian group then the structure tensor $S_G$ of the group algebra is an orbit-indicator of the group $G^2$ identified with $\{(g_1,g_2,g_3)\in G^3: g_1+g_2=g_3\}$; in this case, the tensor $S_G$ has also minimal rank, equal to $|G|$, via the Discrete Fourier Transform. When $G$ is not Abelian, the representation theory of $G$ allows for nontrivial bounds on the rank of $S_G$; similar observations have been leveraged in the study of fast matrix multiplication; e.g.~Cohn and Umans~\cite{CohnU2003}, Cohn, Kleinberg, Szegedy, and Umans~\cite{CohnKSU2005}, Cohn and Umans~\cite{CohnU2012}, Blasiak, Cohn, Grochow, Pratt, and Umans~\cite{BlasiakCGPU2023}. We expect this structure to enable further work towards an eventual proof or disproof of 
the (extended) asymptotic rank conjecture. 

\medskip
\noindent
{\em A win-win dichotomy.}
In addition to enabling worst-case characterisation of tensor exponents for 
families of tensors, the tensors $T^{(g)}$ enable the following 
``win-win'' dichotomy (cf.~Corollary~\ref{cor:dichotomy} for a precise 
statement) that motivates their further study:\\[\medskipamount]
\hspace*{3mm}
\begin{tabular}{r@{\ \ }p{14cm}}
{\em Either} & the extended asymptotic rank conjecture 
(Conjecture~\ref{conj:Strassen-extended}) holds,
implying $\omega=2$ and a disproof of the Set Cover Conjecture;\\
{\em or} & the tensors $T^{(g)}$ form an explicit 
sequence of tensors with superlinear border rank growth, providing 
substantial progress for the ``hay in the haystack'' problem for rank 
and border rank. 
\end{tabular}

\medskip
\noindent
{\em A remark on explicitness.}
We stress that we here view explicitness as the property of not only a single tensor but rather a sequence of tensors; cf.~e.g.~\cite[Section 3]{landsberg2019towards}. Both the tensors $T^{(g)}$ and the tensors in our universal sequences have zero-one entries in coordinates, and these entries may be computed fast. Indeed, from \eqref{eq:t-g} it is immediate that we have linear-time algorithms that output $T_{I,J,K}^{(g)}$ given $I,J,K,g$ as input. Similarly, listing (or ranking/unranking) integer compositions $g\in\sC^{\Delta}_q$ for given $\Delta,q$ admit fast algorithms.

\subsection{Further results}
\label{sect:further-results}

Our second result relates extended Strassen's conjecture to equations of 
varieties, in particular secant varieties. We tacitly assume sufficient 
background in algebraic geometry and geometry of 
tensors~(e.g.~\cite{CoxLO2015,Landsberg2012, michalek2021invitation}).
Also, unless otherwise mentioned, we assume that the field $\F$ is the field
of complex numbers. 

To set the context, it is a well-established method to prove that tensors 
have high border rank by exhibiting polynomials vanishing on the $k$\textsuperscript{th} secant 
variety of the Segre variety $X=(\PP^{d-1})^{\times 3}$ and evaluating it 
on a given tensor $T$. 
Among the state of the art theoretical equations in this context
are the Koszul flattenings 
and the Young flattenings, which can provide border rank bounds up 
$(2-\epsilon)d$ for any $\epsilon>0$ for $d\gg 0$. 
The degree of those equations grows as a polynomial in $d$ for 
fixed $\epsilon$, however the degree of this polynomial also grows 
as $\epsilon\rightarrow 0$. 
The second method to obtain equations of secant varieties is computational, 
based on representation theory and linear algebra. 
This is an exhaustive method, finding all equations in the given degree and partitioning them into so-called isotypic components \cite{breiding2022equations, hauenstein2013equations}. 
Still of course its scope is limited due to computational obstacles.  
For border rank $k<d$, the smallest degree of a polynomial vanishing on 
the $k$\textsuperscript{th} secant variety is exactly $k+1$ \cite{landsberg2004ideals}. 
However, for border rank above $d$, one observes fast grow of the minimal 
degree in which equations exist; for example, the smallest degree of an 
equation vanishing on the $6$\textsuperscript{th} secant variety 
of $(\PP^3)^{\times 3}$ is $19$ and on the $18$\textsuperscript{th} secant variety 
of $(\PP^6)^{\times 3}$ is at least $187 000$ \cite{hauenstein2013equations}. 
The lack of low-degree (or otherwise easy) equations has been perceived so 
far as an obstacle in proving that tensors have high border rank, 
in particular in disproving Strassen's conjecture or its extensions. 
In this paper, we proceed in the {\em opposite} direction, namely that 
{\em absence} of low-degree equations of secant varieties implies upper bound on $\sigma(d)$.

\begin{Thm}[Absence of low-degree equations implies low asymptotic rank]\label{thm:noequlowrank}
Let $X\subseteq \PP((\F^d)^{\otimes 3})$ be a variety contained in the locus of tensors of asymptotic rank at most $r$.
Suppose that no polynomial of degree $p$ vanishes on $X$. Then every tensor in $(\F^d)^{\otimes 3}$ has asymptotic rank at most
\[
r\binom{d^3-1+p}{d^3-1}^{\frac{1}{p}}.
\]
\end{Thm}

We note that the theorem above implies that bounds on asymptotic rank of special tensors may imply bounds for all tensors. As the variety $X$ may be always assumed to be $GL(d)^{\times 3}$ invariant one may combine the computational method of finding isotypic decomposition of homogeneous polynomials with obtaining good bounds on $p$, in order to obtain new bounds on $\sigma(d)$. We leave this line of research for the future. 

\subsection{Related work}

It is known that computing the tensor rank of a given  tensor is NP-hard~\cite{Hastad1990}; see also \cite{HillarL2013,Raz2013,ZhaoWZ2019} for pertinent hardness results. It is also difficult in practice to determine the rank and border rank of small tensors; for example, the rank, border rank and border support rank of $\MM_2$ are known to be seven~\cite{BlaserCZ2018, hauenstein2013equations, Landsberg2006, Strassen1969}. 

There has been extensive interest in constructing explicit tensors of high rank or border rank \cite{alexeev2011tensor, landsberg2015nontriviality, landsberg2019towards, landsberg2018lower}. This study is motivated by the need for new methods to provide lower complexity bounds. 
Currently, we do not know how to construct sequences of explicit tensors $T\in \F^d\otimes\F^d\otimes \F^d$ of rank or border rank above $3d$. In fact, there are no known examples of tensors with entries $0$ or $1$ and rank or border rank greater than $3d$. The only method to construct such tensors, is by making the entries incomparable in size, i.e.~each entry is of different order of magnitude then other entries, or making them algebraically independent, which makes the tensors very far from explicit. However, it is easy to prove that tensors of super-linear border rank, with respect to their size exist. Precisely, for any constant $C>0$ and $d$ sufficiently large there exists a tensor $T\in (\F^d)^{\otimes 3}$ of border rank greater than $Cd$. Even more: general tensors will have greater border rank than $Cd$ for $d$ sufficiently large and the growth of maximal border rank is quadratic in $d$. We simply do not know how to provide explicit examples of such tensors, as the methods we have do not allow to prove that particular tensors have large rank. These obstructions are related to the fact that the cactus variety, that contains the secant variety, fills the whole ambient space, while most of the equations we know for secant varieties, also vanish on cactus varieties \cite{buczynska2014secant, galkazka2023multigraded, galkazka2017vector, iarrobino1999power, landsberg2013equations, teitler2014geometric}. We even do not know how to provide a sequence of explicit tensors so that infinitely many elements of the sequence would have high rank, say above $3d$.
 
The families of tensors with fixed support are also studied. An important concept is that of support rank~\cite{CohnU2012}, which has given rise to other support-based algorithm-design techniques (e.g.~\cite{AlmanZ2023, KarppaK2019}). Further, bounding asymptotic rank of special tensors was recently tied to NP-hard problems; in particular, Strassen's conjecture (Conjecture~\ref{conj:Strassen}) would imply unexpected (but not polynomial) upper bounds on complexity of (randomized) algorithms for NP-hard problems~\cite{BjorklundK2023,Pratt2023}.

\subsection{Organization of this paper}
\label{sect:organization}

Section~\ref{sect:preliminaries} reviews notational preliminaries and 
definitions.
Section~\ref{sect:universal-sequences} studies the composition basis and proves
our main theorems on explicit universal sequences of tensors
(Theorems~\ref{thm:main}~to~\ref{thm:extended-universal}). 
Section~\ref{sect:secant-equations} relates the equations of varieties
to bounds on asymptotic rank and proves Theorem~\ref{thm:noequlowrank}.

\section{Preliminaries}

\label{sect:preliminaries}

Recall that we write $[d]=\{1,2,\ldots,d\}$. 
The set $[d]^q$ consists of sequences of length $q$ of integers in $[d]$. We fix the canonical basis
$e_1,e_2,\ldots,e_m\in\F^m$.

In this article, we exclusively work with tensors of format $a\times b\times c$, where in most cases $a=b=c$. These are elements of the vector space $\F^a\otimes\F^b\otimes \F^c\simeq \F^{a\times b\times c}$. In analogy to the case of matrices, the reader may freely think about tensors as three dimensional arrays filled with elements of $\F$. For a vector $v\in \F^a$ we write $v_i:=e_i^*(v)\in \F$ for $1\leq i\leq a$. We write $S_{i,j,k}\in\F$ for the entry of $S\in \F^a\otimes\F^b\otimes \F^c$ at position 
$(i,j,k)\in[a]\times[b]\times[c]$.

For three vectors $v_1\in\F^a, v_2\in \F^b$ and $v_3\in \F^c$ we define the tensor $v_1\ot v_2\ot v_3\in \F^a\otimes\F^b\otimes \F^c$ where the $(i,j,k)$ coordinate equals $(v_1)_i (v_2)_j (v_3)_k$. Tensors of this form are called \emph{rank one} tensors. The \emph{rank} of a tensor $T$ is the smallest $r$ such that $T$ is sum of $r$ rank one tensors.

In case $\F=\R$ or $\F=\C$ we define the \emph{border rank} of a tensor $T$ as the smallest $r$ such that in any neighbourhood of $T$ there exists a tensor of rank $r$. For more details about rank and border rank we refer to \cite{burgisser2013algebraic, Landsberg2012, michalek2021invitation}. The Kronecker power of a tensor is defined as in formula \eqref{eq:box-entry}. \emph{Asymptotic rank}~\cite{Gartenberg1985} of a tensor $T$ is defined 
as $\lim_{n\rightarrow\infty} \Ra(T^{\boxtimes n})^{\frac{1}{n}}$. 
We write $\langle s\rangle$ to be the unit tensor $\langle s\rangle=\sum_{i=1}^s e_i\ot e_i\ot e_i$.

\section{Universal sequences of tensors for the asymptotic rank conjecture}

\label{sect:universal-sequences}

In this section we prove that the tensors $T^{(g)}$ in the composition basis
(Definition~\ref{def:Tg})
form a universal family for the asymptotic rank conjecture. 
We start with a more detailed analysis of the structure of the composition 
basis for fixed $d$ (Section~\ref{sect:comp-properties}), and 
follow with an analysis for increasing~$d$ (Section~\ref{sect:limit-exponent});
in particular, we establish the existence of the limit 
exponent $\lim_{d\rightarrow\infty}\sigma(d)$. 
Theorem~\ref{thm:main}~and~Theorem~\ref{thm:extended-universal} are 
restated and proved next (Section~\ref{sect:universal-seq-proofs}).
Using techniques of Strassen, we then show that the universal sequences
have nontrivially low tensor rank (Section~\ref{sect:upper-bounds-on-rank}).
We end this section by setting up our techniques for support-localization
(Section~\ref{sect:support-local}) as well as restate and prove 
Theorem~\ref{thm:tight-universal}, our main result for tight tensors 
(Section~\ref{sect:tight}).

\subsection{Properties of the composition basis}
\label{sect:comp-properties}

We start with the following easy lemma that provides various alternative
characterisations of tensors $T^{(g)}$ in Definition~\ref{def:Tg}.

We note that the symmetric group $\SS_q$ acts by permutations 
on $\F^{[d]^q}$ and diagonally on $(\F^{[d]^q})^{\otimes 3}$.

\begin{Lem}[Equivalent definitions of the composition basis]\label{lem:kdp}
We have the following equivalent definitions of tensors $T^{(g)}$ for any $g\in\sC_q^{[d]^3}$.
\begin{enumerate}
\item The tensors $T^{(g)}$ are precisely sums of $\SS_q$ orbits of canonical basis tensors in $(\F^{[d]^q})^{\otimes 3}$.
\item The tensors $T^{(g)}$ are coefficients of monomials for the Kronecker map $K_{d,q}$, explicitly:
\[
K_{d,q}(S)=\sum_{g\in\sC_q^{[d]^3}}S^g T^{(g)}\,.
\]
In particular, the linear span of the image is contained in the linear span of tensors $T^{(g)}$.
\item If\/ $|\F|>q$, then the tensors $T^{(g)}$, up to scaling by a constant, are precisely tensors with inclusion minimal supports in the linear span of the image of $K_{d,q}$. 
\end{enumerate}
\end{Lem}
\begin{proof}
We note that the condition $\Phi_{I,J,K}=g$ from Definition \ref{def:Tg} determines the triple $(I,J,K)$ up to simultaneous action by an element of $\SS_q$. This proves equivalence of point $(1)$ with the original definition. 

The coordinates of $K_{d,p}(S)$ in the canonical basis are monomials of degree $q$ in the coordinates of $S$. A monomial $S^g$ appears exactly on the entries indexed by such $(I,J,K)$'s that $\Phi_{I,J,K}=g$. This proves equivalence of point $(2)$ with Definition \ref{def:Tg}.

In Proposition \ref{prop:T^ginimage}, we show that the linear span of the image of $K_{d,q}$ coincides with the linear span of $T^{(g)}$'s for $g\in \sC_q^{[d]^3}$ and $|\F|>q$. As $T^{(g)}$'s have disjoint supports, we obtain point (3).
\end{proof}

\begin{Rem}[Invariant subspaces defined by marginals of $g$] It is immediate that $T^{(g)}\in (\F^{[d]^q})^{\otimes 3}$, however the tensor $T^{(g)}$ also belongs to smaller invariant subspace defined by $g$. Namely, for $g\in \sC_q^{[d]^3}$ let $g_1:[d]\rightarrow \N$ be the first marginal of $g$; that is, let $g_1(j)=\sum_{a,b\in[d]} g(j,a,b)$ for all $j\in[d]$. In the same way, define the second marginal $g_2$ and third marginal $g_3$. For $i=1,2,3$, let $U_i\subseteq[d]^q$ be the set of all $q$-tuples $I\in [d]^q$ such that the value $j$ appears in $I$ exactly $g_i(j)$ times for all $j\in [d]$. We have $|U_i|=\binom{q}{g_i(1),g_i(2),\dots,g_i(d)}$ as well as $T^{(g)}\in \F^{U_1}\otimes\F^{U_2}\otimes \F^{U_3}$. Clearly, the ambient space is $\SS_q$ invariant. 
\end{Rem}

\begin{Exa}[The small Coppersmith--Winograd tensor]
Let $d=2$, $q=3$ and $g$ assign value $1$ on $(0,0,1)$, $(0,1,0)$ and $(1,0,0)$. The tensor $T^{(g)}$ is the small Coppersmith--Winograd tensor in $(\F^3)^{\otimes 3}$, that is:
\[
e_0\otimes e_1\otimes e_2+e_0\otimes e_2\otimes e_1+e_1\otimes e_0\otimes e_2+e_1\otimes e_2\otimes e_0+e_2\otimes e_0\otimes e_1+e_2\otimes e_1\otimes e_0\,.
\]
\end{Exa}

\begin{Def}[Linear span of the composition basis]
Let $L_{d,q}$ be the linear span of the tensors $T^{(g)}$ for $g\in\sC_q^{[d]^3}$. As the tensors $T^{(g)}$ have disjoint supports they are linearly independent and hence form a basis of $L_{d,q}$. 
We note that $\dim L_{d,q}=\binom{d^3+q-1}{q}$ which is the cardinality of $\sC_q^{[d]^3}$.
\end{Def}
\begin{Cor}[Span coincides with the space of invariants]\label{cor:tgbasisofL}
The tensors $T^{(g)}$ for $g\in \sC_q^{[d]^3} $ form a basis of the invariants space $((\F^{[d]^q})^{\otimes 3})^{\SS_q}$, which coincides with $L_{d,q}$. Up to rescaling, it is the unique basis of that space made of tensors with disjoint supports. 
\end{Cor}
From now on we will assume that the cardinality of the field is greater than $q$.
\begin{Lem}[Dual space of homogeneous polynomials]\label{lem:lintopoly}
There is an isomorphism of vector spaces: $L_{d,q}^*$ of linear forms and functions on $(\F^d)^{\ot 3}$ given by homogeneous polynomials of degree $q$. The isomorphism sends a linear form $l$ to the polynomial function $l\circ K_{d,q}$. 
\end{Lem}
\begin{proof}
Let $T^{(g)*}$ be the basis of $L_{d,q}^*$ dual to $T^{(g)}$.
The image of the linear function $\sum \lambda_g  T^{(g)*}$ is $\sum \lambda_g S^g$. Surjectivity is obvious. 

By induction on the number $n$ of variables, one proves that no polynomial of degree $q$ may vanish identically on $\F^n$, when $|\F|>q$. 
Hence, no linear form $l$ is mapped to the zero function and the linear map $l\mapsto l\circ K_{d,p}$ is injective. This finishes the proof. 
\end{proof}
\begin{Prop}[Span of the image of the Kronecker power map]\label{prop:T^ginimage}
The space $L_{d,q}$ is the linear span of the image of $K_{d,q}$.
\end{Prop}
\begin{proof}
Clearly $L_{d,q}$ contains the image of $K_{d,q}$. If the containment was strict, there would exist a nonzero linear function $l\in L_{d,q}$ vanishing on the image. Then $l\circ K_{d,q}=0$. This is not possible by Lemma \ref{lem:lintopoly}.
\end{proof}
\begin{Rem}[The Veronese map of degree $q$]
Up to isomorphism, $K_{d,q}$ may be identified with the $q$\textsuperscript{th} Veronese map, that is a map defined by all degree $q$ monomials. However, in coordinates each monomial may appear more than once, as each monomial represented by $g\in \sC_q^{[d]^3}$ appears on the support of $T^{(g)}$. 
\end{Rem} 

\begin{Lem}[Maximum rank in the composition basis controls rank in $L_{d,q}$]\label{lem:estimate_from_Tg}
Let $r$ be the maximum rank (respectively, asymptotic rank) of $T^{(g)}$ over $g\in\sC_q^{[d]^3}$.
Every tensor in $L_{d,q}$ has rank (respectively, asymptotic rank) at most
\[
r|\sC_q^{[d]^3}|=r\binom{d^3-1+q}{d^3-1}\,.
\]
In particular, every tensor in $(\F^d)^{\ot 3}$ has asymptotic rank at most 
\[
\left(r\binom{d^3-1+q}{d^3-1}\right)^{\frac{1}{q}}\,.
\]
\end{Lem}
\begin{proof}
Fix $T \in (\F^d)^{\ot 3}$.
By Lemma \ref{lem:kdp} the tensor $K_{d,q}(T)$ is a linear combination of $T^{(g)}$'s for $g\in \sC_q^{[d]^3}$. As $|\sC_q^{[d]^3}|=\binom{d^3-1+q}{d^3-1}$ we obtain:
\[
\Ra (K_{d,q}(T))\leq r\binom{d^3-1+q}{d^3-1}\,.
\]
The statement under the assumption of asymptotic rank $r$ for $T^{(g)}$'s is proved in the same way, noting that asymptotic rank is subadditive and submultiplicative (e.g.~\cite{WigdersonZ2023}).
\end{proof}

\subsection{The extended asymptotic rank conjecture and the limit exponent}
\label{sect:limit-exponent}

We are now ready to connect the composition-basis tensors $T^{(g)}$ 
to the extended
asymptotic rank conjecture (Conjecture~\ref{conj:Strassen-extended}).

\begin{Cor}[The composition basis suffices for the extended asymptotic rank conjecture]
If there exists an infinite set $S=\{d_1,d_2,\ldots\}\subseteq \Z_{\geq 1}$ such that for any $d_i$ there exist infinitely many $q_{i,j}\in\Z_{\geq 1}$ such that Conjecture~\ref{conj:Strassen-extended} holds for the tensor $T^{(g)}$ for all $g\in\sC_{q_{i,j}}^{[d_i]^3}$, then Conjecture~\ref{conj:Strassen-extended} holds for all tensors. 
\end{Cor}
\begin{proof}
First fix $d_i$. By Lemma \ref{lem:estimate_from_Tg} we see that any $T\in ((\F^{d_i})^{\ot 3})$ has asymptotic rank at most: 
\[
(d_i^q)^{\frac{1}{q}}\binom{d_i^3-1+q}{d_i^3-1}^{\frac{1}{q}}\,.
\]
Note that for fixed $d_i$ we have 
\[
\lim_{q\rightarrow\infty} \binom{d_i^3-1+q}{d_i^3-1}^{\frac{1}{q}}=1\,.
\]
This confirms Conjecture \ref{conj:Strassen-extended} for any $T\in (\F^{d_i})^{\ot 3}$. We conclude that it must hold for each $d$ from the next Lemma \ref{lem:forbigthenforall}.
\end{proof}
\begin{Lem}[Existence of the limit exponent]\label{lem:forbigthenforall}
The limit\/ $\lim_{d\rightarrow \infty} \sigma(d)$ exists and equals the supremum of the set\/ $\{\sigma(d):d\in\Z_{\geq 1}\}$.
\end{Lem}
\begin{proof}
As the sequence $\sigma(d)$ is bounded it is enough to prove:
\[
\forall_{d_0\in\Z_{\geq 1}}\forall_{\epsilon>0}\exists_{D\in\Z_{\geq 1}}\forall_{d>D}\quad \sigma(d)>\sigma(d_0)-\epsilon\,.
\]

Fix $\epsilon>0$ and $d_0\in\Z_{\geq 1}$. For contradiction assume there are infinitely many $d_i$ such that $\sigma(d_i)\leq \sigma(d_0)-\epsilon$. Hence, we also obtain infinitely many $k_i$ such that $k_id_0\leq d_i< (k_i+1)d_0$.

Fix $T\in (\F^{d_0})^{\otimes 3}$ with $\sigma(T)>\sigma(d_0)-\epsilon/2$.
The direct sum $T^{\bigoplus k_i}$ of $k_i$ tensors $T$ has asymptotic rank at most $d_i^{\sigma(d_0)-\epsilon}<(d_0(k_i+1))^{\sigma(d_0)-\epsilon}$. 
Let $b=d_0^{\sigma(T)}$ be the asymptotic rank of $T$. 
For $\delta>0$ and $k_i\gg 0$ we may degenerate the unit tensor $\langle k_i\rangle$ to $T^{\boxtimes \log_{b-\delta}k_i}$ and $T^{\oplus k_i}=\langle k_i\rangle\boxtimes T$ to $T^{\boxtimes 1+\log_{b-\delta}k_i}$. 
This implies $b^{1+\log_{b-\delta}k_i}\leq (d(k_i+1))^{\sigma(d_0)-\epsilon}$. Taking the limit as $k_i\rightarrow\infty$ we obtain a contradiction with $b\geq d_0^{\sigma(d_0)-\epsilon/2}$.   
\end{proof}
The following more precise result shows that the exponent governing the asymptotic rank is completely determined via asymptotic ranks of the tensors $T^{(g)}$.
\begin{Cor}[The composition bases govern tensor exponents]
For all $\tau\geq 0$ it holds that 
each tensor $T^{(g)}\in (\F^m)^{\otimes 3}$ (where $m$ is uniquely determined by $g$) has asymptotic rank at most $m^\tau$ if and only if any tensor $T\in (\F^d)^{\otimes 3}$ (where $d$ is arbitrary) has asymptotic rank at most $d^\tau$. In particular, for any tensor $T$ we have $\sigma(T)\leq \sup_g \sigma(T^{(g)})$.
\end{Cor}
\begin{proof}
The implication ``$\Leftarrow$'' is immediate. For the implication ``$\Rightarrow$'', assume that a tensor $T$ has asymptotic rank greater than $d^{\tau+\epsilon}$ for fixed $\epsilon >0$. By applying Lemma \ref{lem:estimate_from_Tg} for $q\gg 0$ such that $\binom{d^3-1+q}{d^3-1}^{\frac{1}{q}}<d^\epsilon$ we obtain a contradiction.
\end{proof}

We note that there are countably many tensors $T^{(g)}$ and they form a sequence of explicit tensors. This sequence gives rise to the following dichotomy.

\begin{Cor}[A dichotomy on improved algorithms or explicit rank lower bounds]
\label{cor:dichotomy}
We have the following dichotomy. Either
\begin{enumerate}
\item[(i)] the extended asymptotic rank conjecture (Conjecture~\ref{conj:Strassen-extended}) holds; that is, for all tensors $T$ we have $\sigma(T)\leq 1$, which implies $\omega=2$ and disproves the Set Cover Conjecture; or
\item[(ii)] the tensors $T^{(g)}$ form a sequence of explicit tensors such that for any constant $C\in \Z_{\geq 1}$ there exist infinitely many $T^{(g)}\in (\F^m)^{\otimes 3}$, such that rank, border rank, and asymptotic rank is greater than $Cm$.
\end{enumerate}
\end{Cor}

It is clear that for special $g$ the rank of $T^{(g)}$ can be small. For example if $g$ is supported at one element, then $T^{(g)}$ has support of cardinality one and thus is of rank one.

\subsection{Explicit universal sequences and main results}
\label{sect:universal-seq-proofs}

We first prove our main theorem for fixed~$d$, and define 
the universal sequence accordingly. 

\begin{Def}[Universal sequence for a fixed $d$]\label{def:Ud}
For a fixed positive integer $d$, we define the sequence 
$\cU_d=(U_{d,q}:q=1,2,\ldots)$ of tensors, where we set
\[
U_{d,q}=\bigoplus_{g\in \sC_{q}^{[d]^3} } T^{(g)}\,.
\]
\end{Def}
We are now ready to prove our main theorem; Theorem~\ref{thm:main} is 
restated below for convenience. 

\begin{Thm}[Main; Explicit universal sequences of tensors]\label{thm:mainintext}
For all $d\geq 1$ we have $\sigma(\cU_d)=\sigma(d)$. 
\end{Thm}
\begin{proof}
We prove the equality by proving both inequalities. 

First, suppose that $\sigma(\cU_{d})< \sigma$ for some constant $\sigma$. Fix an $\epsilon>0$ so that for $q$ sufficiently large $\Ra(U_{d,q}) < (|\sC_{q}^{[d]^3}|d^q)^{\sigma-\epsilon}$. Then, for $q$ sufficiently large and all $g\in \sC_{q}^{[d]^3}$ we have $\Ra(T^{(g)})\leq\Ra(U_{d,q}) < (|\sC_{q}^{[d]^3}|d^q)^{\sigma-\epsilon} $. As $\lim_{q\rightarrow\infty}  |\sC_{q}^{[d]^3}|^{1/q}=1$ we see that for $q$ sufficiently large and all $g\in \sC_{q}^{[d]^3}$ we have $\sigma(T^{(g)})\leq\sigma$. Applying Lemma \ref{lem:estimate_from_Tg} and taking limit, as $q\rightarrow\infty$ we see that $\sigma(T)\leq \sigma$ for all $T\in (\F^d)^{\otimes 3}$.

Suppose now that for all $T\in (\F^d)^{\otimes 3}$ we have $\sigma(T)<\sigma$. Fix $\epsilon>0$, such that $\sigma(T)<\sigma-\epsilon$ for all $T\in (\F^d)^{\otimes 3}$. By Proposition \ref{prop:T^ginimage} for any $g\in \sC_{q}^{[d]^3}$ the tensor $T^{(g)}$ is a sum of at most $|\sC_{q}^{[d]^3}|$-many tensors of type $K_{d,q}(T)$. Hence, for $q$ sufficiently large we have:
\[
\Ra(T^{(g)})\leq |\sC_{q}^{[d]^3}|d^{q(\sigma-\epsilon)}\,.
\]
It follows that:
\[
\Ra(U_{d,q})\leq |\sC_{q}^{[d]^3}|^2d^{q(\sigma-\epsilon)}\,.
\]
As  $\lim_{q\rightarrow\infty} |\sC_{q}^{[d]^3}|/d^{\epsilon q}=0$ we see that $\Ra(U_{d,q})\leq (|\sC_{q}^{[d]^3}| d^q)^\sigma$ for $q\gg 0$. Thus
$\sigma(\cU_{d})\leq \sigma$, which finishes the proof.
\end{proof}

Next, we construct one universal sequence realizing the worst possible exponent, irrespective of $d$ via a diagonal argument on the sequences $\cU_d$. 
Accordingly, recalling Definition~\ref{def:Ud}, define the diagonal 
sequence $\cD=(D_d:d=1,2,\ldots)$ of tensors 
for all positive integers $d$ by setting $D_d=U_{d,d^4}$. 
It is immediate that the sequence is explicit. 
Theorem~\ref{thm:extended-universal} is restated below for convenience. 

\begin{Thm}[An explicit universal sequence for the extended asymptotic rank conjecture]\label{thm:extended-universalintext}
The sequence $\cD$ satisfies $\lim_{d\rightarrow\infty}\sigma(D_d)=\lim_{d\rightarrow\infty}\sigma(d)$.
\end{Thm}
\begin{proof}
By Lemma \ref{lem:forbigthenforall} it is enough to prove that for every $\epsilon>0$, for all $d$ sufficiently large we have $\sigma(U_{d,d^4})>\sigma(d)-\epsilon$. 
For any (concise) $T\in (\F^d)^{\otimes 3}$ we have:
\[
d^{d^4\sigma(T)}=\ARa(T)^{d^4}\leq \sum_{g\in \sC_{d^4}^{[d]^3}}\ARa(T^{(g)})\leq |\sC_{d^4}^{[d]^3}|\ARa(U_{d,d^4})\leq |\sC_{d^4}^{[d]^3}|(|\sC_{d^4}^{[d]^3}|d^{d^4})^{\sigma(U_{d,d^4})}\,.
\]

After taking root of order $d^4$ we only need to prove that $\lim_{d\rightarrow\infty}|\sC_{d^4}^{[d]^3}|^{1/d^4}=1$. We have:
\[
\lim_{d\rightarrow\infty}|\sC_{d^4}^{[d]^3}|^{1/d^4}=\lim_{d\rightarrow\infty}\binom{d^4+d^3-1}{d^3-1}^{1/d^4}\leq \lim_{d\rightarrow\infty}(2d^4)^{d^3/d^4}=1\,,
\]
which finishes the proof.
\end{proof}

\subsection{Upper bounds on tensor rank in the universal sequences}
\label{sect:upper-bounds-on-rank}

A natural question arises: what are ranks of the tensors $T^{(g)}$. 
In order to provide upper bounds, we recall the following lemma due to Strassen
(implicit in \cite[Proposition~3.6]{Strassen1988}). 
\begin{Lem}[Limited universality of matrix multiplication tensors]\label{lem:strassen}
The matrix multiplication tensor $\MM_{d^2}$ degenerates to $T^{\boxtimes 3}$ for any tensor $T\in (\F^d)^{\otimes 3}$. In particular, $\sigma(T)\leq 2\omega/3$.
\end{Lem}
\begin{proof}
Let us write in coordinates $T=(\lambda_{p,q,r})$ and 
\[
\MM_{d^2}=\sum_{i_1,i_2,j_1,j_2,k_1,k_2\in[d]} e^{i_1,i_2}_{j_1,j_2}\ot e^{j_1,j_2}_{k_1,k_2}\ot e^{k_1,k_2}_{i_1,i_2}\,.
\]
For $i=1,2,3$, define $l_i:\F^{d^4}\rightarrow \F^{d^3}$ as follows:
\begin{align*}
l_1(e^{i_1,i_2}_{j_1,j_2})=\sum_{a\in[d]} \lambda_{a, j_1, i_1} e_{a, j_2, i_2}\,,\quad l_2(e^{j_1,j_2}_{k_1,k_2})=\sum_{b\in[d]} \lambda_{j_2, b, k_1} e_{j_1, b, k_2}\,, \quad l_3(e^{k_1,k_2}_{i_1,i_2})=\sum_{c\in[d]} \lambda_{i_2, k_2, c} e_{i_1,k_1 , c}\,.
\end{align*}
We have $l_1\ot l_2\ot l_3(\MM_{d^2})=T^{\boxtimes 3}.$
\end{proof}
In a similar way we obtain:
\begin{Lem}[Limited universality of matrix multiplication under powers]\label{lem:Tgasympdeg}
The matrix multiplication tensor $\MM_{d^{2k}}$ degenerates to $T^{\boxtimes {3k}}$ for any tensor $T\in (\F^d)^{\otimes 3}$.
\end{Lem}
We obtain the following corollary, which currently is asymptotically in $k\gg 0$ the best one we know, that works for arbitrary $g\in \sC_{3k}^{[d]^3}$.
\begin{Cor}[Upper bound on rank via matrix multiplication]\label{cor:TgasDegofM}
For any  $g\in \sC_{3k}^{[d]^3}$ the tensor $T^{(g)}$ has rank at most $\Ra(\MM_{d^{2k}})\cdot\binom{d^3-1+3k}{d^3-1}.$
\end{Cor}
\begin{Rem}[Note on special $g$]
We note that Corollary \ref{cor:TgasDegofM} may be improved for many special~$g$. For any $g\in \sC_{3k}^{[d]^3}$ we have three marginal distributions on $[d]$. If such a distribution is close to uniform, then we obtain about $d^k$ compatible sequences of length $k$; that is, asymptotically every sequence is compatible. However, if the distribution is far from uniform, we obtain less sequences, say $d^{ck}$ for some $c<1$. In such a case $\MM_{d^{2k}}$ in Lemma~\ref{lem:Tgasympdeg} may be changed into smaller matrix multiplication. 
\end{Rem}

\subsection{Support-localized universality}
\label{sect:support-local}

Many of our results remain true if we consider tensors $T\in (\F^d)^{\otimes 3}$ with support contained in a nonempty $\Delta\subseteq [d]^3$. Let us write $\F^\Delta$ be the linear space of tensors with such a support. As before $K_{d,q}$ restricted to $\F^\Delta$ is still a Veronese map, simply in variables corresponding to elements from $\Delta$. The image of $K_{d,q}|_{\F^\Delta}$ has a basis made of tensors $T^{(g)}$ for $g\in \sC^\Delta_q$. Below we present a variant of Lemma \ref{lem:estimate_from_Tg}, which proof is analogous.
\begin{Lem}[Maximum rank in the support-localized composition basis controls rank in $\F^\Delta$]\label{lem:estimate_from_Tg2}
Let $r$ be the maximum rank (respectively, asymptotic rank) of $T^{(g)}$ over $g\in\sC_q^{\Delta}$.
The $q$\textsuperscript{th} Kronecker power of every tensor in $\F^\Delta$ has rank (respectively, asymptotic rank) at most
\[
r|\sC_q^{\Delta}|=r\binom{|\Delta|-1+q}{|\Delta|-1}\,.
\]
In particular, every tensor in $(\F^d)^{\ot 3}$ with support contained in $\Delta$ has asymptotic rank at most 
\[
\left(r\binom{|\Delta|-1+q}{|\Delta|-1}\right)^{\frac{1}{q}}\,.
\]
\end{Lem}
Next, we present a variant of Strassen's result for tensors with fixed support. For $\Delta\subseteq[d]^3$, we will use the following notation: $\Delta_{1,2}=\{(i,j):\exists_k (i,j,k)\in \Delta\}$, $\Delta_{1,3}=\{(i,k):\exists_j (i,j,k)\in \Delta\}$, and $\Delta_{2,3}=\{(j,k):\exists_i (i,j,k)\in \Delta\}$.

Consider the following support-localized
restriction of the matrix multiplication tensor $\MM_{d^2}$:
\[
\MM_\Delta:=\sum_{\substack{(j_1,i_1)\in\Delta_{2,3}\\(j_2,k_1)\in \Delta_{1,3}\\(i_2,k_2)\in\Delta_{1,2}}} e^{i_1,i_2}_{j_1,j_2}\ot e^{j_1,j_2}_{k_1,k_2}\ot e^{k_1,k_2}_{i_1,i_2}\,.
\]
We obtain the following version of Lemma \ref{lem:strassen}.
\begin{Lem}[Limited universality of support-localized matrix multiplication]
The tensor $\MM_{\Delta}$ degenerates to $T^{\boxtimes 3}$ for any tensor $T\in (\F^d)^{\otimes 3}$ with support contained in $\Delta$.
\end{Lem}
\begin{Cor}[Upper bound on rank via support-localized matrix multiplication]
For any  $g\in \sC_{3k}^{\Delta}$ the tensor $T^{(g)}$ has rank at most $\Ra(\MM_{\Delta})\cdot\binom{|\Delta|-1+3k}{|\Delta|-1}.$
\end{Cor}
Next we present the proof of Theorem \ref{thm:main-localized}.
\begin{Thm}[Support-localized universal sequences of tensors]\label{thm:main-localizedintext}
For all nonempty $\Delta\subseteq[d]\times[d]\times[d]$ there is 
an explicit sequence $\cU_\Delta$ of zero-one-valued tensors 
with $\sigma(\cU_\Delta)=\sigma(\Delta)$. 
\end{Thm}
\begin{proof}
Let $\cU_\Delta=(U_{\Delta,q}:q=1,2,\dots)$ for
\[U_{\Delta,q}=\bigoplus_{g\in\sC^\Delta_q}T^{(g)}.\]

As the image of $K_{d,q}|_{\F^\Delta}$ has a basis made of tensors $T^{(g)}$ for $g\in \sC^\Delta_q$ the proof is analogous to the proof of Theorem \ref{thm:mainintext} where we replace reference to Lemma \ref{lem:estimate_from_Tg} by reference to Lemma \ref{lem:estimate_from_Tg2}.
\end{proof}

\subsection{Tight tensors and Strassen's conjecture}
\label{sect:tight}

We now proceed to our main theorem for tight tensors and 
fixed~$d$; we start by defining a corresponding universal sequence. 

\begin{Def}[Universal sequence for tight tensors and fixed $d$]\label{def:Td}
For a fixed positive integer $d$, we define the sequence 
$\cT_d=(T_{d,q}:q=1,2,\ldots)$ of tensors, where we set
\[
T_{d,q}=\bigoplus_{\substack{\Delta\subseteq [d]^3\\\text{$\Delta$ tight}}}U_{\Delta,q}=\bigoplus_{\substack{\Delta\subseteq [d]^3\\\text{$\Delta$ tight}}}\bigoplus_{g\in \sC_{q}^{\Delta} } T^{(g)}\,.
\]
\end{Def}
Next we prove Theorem~\ref{thm:tight-universal}, which we restate below 
for convenience. 
\begin{Thm}[A universal sequence of tensors for Strassen's conjecture]\label{thm:tight-universalintext}
We have $\sigma(\cT_d)=\sigma(\sT_d)$ and each $T_{d,q}$ is tight.
\end{Thm}
\begin{proof}
First, we prove that for $g\in \sC_{q}^{\Delta}$, where $\Delta$ is tight, the tensor $T^{(g)}$ is tight. Indeed, if we take any tensor $T$ with support $\Delta$, it is tight, thus so is $T^{\boxtimes q}$. As $T^{(g)}$ has smaller support it is tight. As direct sum of tight tensors is tight, we see that each $T_{d,q}$ is tight. 

Next we note that for fixed $d$ each $T_{d,q}$ is a sum of a fixed number of tensors $U_{\Delta,q}$. Thus, there exists a constant $C$ such that:
\[\max_{\substack{\Delta\subseteq [d]^3\\\text{$\Delta$ tight}}}\Ra (U_{\Delta,q})\leq \Ra (T_{d,q})\leq C \max_{\substack{\Delta\subseteq [d]^3\\\text{$\Delta$ tight}}}\Ra (U_{\Delta,q})\]
and the same inequalities hold for the size of the tensors. For every tight $T$, we may find tight $\Delta$ such that $\sigma(T)\leq \sigma(\Delta)$. By Theorem \ref{thm:main-localized} this equals $\sigma(\cU_\Delta)$ and by the inequality above this is upper bounded by $\sigma(\cT_d)$. 

For the other inequality, we see that $\sigma(\cT_d)$ is upper bounded by $\sigma(\cU_\Delta)$ for some tight $\Delta$. Clearly, $\sigma(\cU_\Delta)\leq \sigma(\sT_d)$ which finishes the proof. 
\end{proof}

\section{Equations of secant varieties and asymptotic rank}

\label{sect:secant-equations}

In this section we relate existence of polynomials vanishing on special varieties, like secant varieties, and asymptotic rank. In particular, we obtain upper-bound control on $\sigma(d)$ via the {\em absence} of low-degree equations 
of secants. 

\subsection{Equations and bounds on $\sigma(d)$}

We start with a generalization of Proposition \ref{prop:T^ginimage}.

\begin{Prop}[Span of the image of a subset]\label{prop:setspanning}
For any subset $S\subseteq (\F^d)^{\otimes 3}$, the following are equivalent:
\begin{enumerate}
\item no homogeneous polynomial of degree $p$ vanishes on $S$,
\item $K_{d,p}(S)$ linearly spans $L_{d,p}$.
\end{enumerate}
\end{Prop}
\begin{proof}
The second condition is equivalent to the fact that no linear form in $L_{d,p}^*$ vanishes on $K_{d,p}(S)$. By Lemma \ref{lem:lintopoly} this is also equivalent to the first condition.
\end{proof}
\begin{Cor}[Absence of low-degree equations implies low asymptotic rank]\label{cor:noeqthenlowrank}
Let $X$ be the Segre variety of rank one tensors in $(\F^d)^{\ot 3}$. Suppose that for fixed $n$ no polynomial of degree $p$ vanishes on the $n$\textsuperscript{th} secant variety $\sigma_n(X)$. Then every tensor in $L_{d,p}$ has rank at most
\[
\binom{d^3-1+p}{d^3-1}n^p\,.
\]
Every tensor in $(\F^d)^{\ot 3}$ has asymptotic rank at most
\[
n\binom{d^3-1+p}{d^3-1}^{\frac{1}{p}}\,.
\]
\end{Cor}
\begin{proof}
Note that  $\dim L_{d,p}=\binom{d^3-1+p}{d^3-1}$ and by Proposition \ref{prop:setspanning} we have a basis of $L_{d,p}$ made of tensors of rank at most $n^p$.
The last statement follows. 
\end{proof}
A more general, but less explicit version is given by Theorem \ref{thm:noequlowrank}, which we derive as the following corollary.
\begin{Cor}[Absence of low-degree equations implies low asymptotic rank; implicit version]\label{cor:noeqthenlowrankimplicit}
Let $Y\subseteq (\F^d)^{\ot 3}$ be a subset with the property that for all $T\in Y$ the asymptotic rank of $T$ is at most $n$. Suppose that no homogeneous polynomial of degree $p$ vanishes on $Y$. Then, every tensor in $(\F^d)^{\ot 3}$ has asymptotic rank at most
\[
n\binom{d^3-1+p}{d^3-1}^{\frac{1}{p}}\,.
\]
\end{Cor}

\begin{Rem}\label{rem:improvements}
The dimension component $\binom{d^3-1+p}{d^3-1}$ may be improved (by going to border rank) by the smallest $s$, such that $\sigma_s(K_{d,p}(\sigma_n(X)))=L_{d,p}$. The latter number is in most cases much smaller then $\dim L_{d,p}$. 
Each particular case may be exactly studied by the Teraccini Lemma. 

Alternatively we could greedily, starting from rank one and going up, take tensors $T_i$ such that $K_{d,p}(T_i)$ are linearly independent, building a basis of $L_{d,p}$. 
\end{Rem}

For tensors belonging to a fixed subspace $W\subseteq (\F^d)^{\otimes 3}$ we also have the following variant.
\begin{Thm}
Let $Y\subseteq W\subseteq (\F^d)^{\otimes 3}$ be a subset with the property that for all $T\in Y$ the asymptotic rank of $T$ is at most $n$. Suppose that no homogeneous polynomial on $W$ of degree $p$ vanishes on $Y$. Then every tensor in $W$ has asymptotic rank at most
\[
n\binom{\dim W-1+p}{\dim W-1}^{\frac{1}{p}}\,.
\]
\end{Thm}
\begin{proof}
We recall that the restriction of the Veronese map of degree $p$ to any linear subspace is still a Veronese map of degree $p$; that is, ~a map defined by polynomials spanning the space of degree $p$ polynomials, simply in smaller number of variables. 

We restrict $K_{d,p}$ to $W$. The linear span of $K_{d,p}(W)$ is spanned by $K_{d,p}(Y)$ and has dimension $\binom{\dim W-1+p}{\dim W-1}$. Hence for any tensor $T\in W$ the asymptotic rank of $K_{d,p}(T)$ is at most:
\[
n^p\binom{\dim W-1+p}{\dim W-1}\,.
\]
As $K_{d,p}(T)$ is the $p$\textsuperscript{th} Kronecker power of $T$, the statement follows. 
\end{proof}  
The previous result could be particularly useful if we can provide many examples of low rank tensors inside a linear subspace. 

\subsection{Equations of secant varieties and applications}
We note that deciding if there exists a nonzero polynomial of degree $p$ on $(\F^d)^{\otimes 3}$ vanishing on tensors of rank $r$ is a question in linear algebra. The na\"\i{}ve approach is a as follows.

Consider $3r$ vectors $v_{i,j}\in \F^d$ where $1\leq i \leq 3$ and $1\leq j\leq r$ with coordinates $(v_{i,j})_k$ treated as variables. The tensor $T:=\sum_{j=1}^r v_{1,j}\otimes v_{2,j}\otimes v_{3,j}$ is a general tensor of rank $r$. We have:
\[
T_{a,b,c}=\sum_{j=1}^r (v_{1,j})_a (v_{2,j})_b (v_{3,j})_c\,.
\]

If we assign the weight $(1,0,0)\in\Z^3$ to each $(v_{1,j})_a$, the weight $(0,1,0)$ to each $(v_{2,j})_b$, and the weight $(0,0,1)$ to each $(v_{3,j})_c$, we see that an evaluation of a degree $p$ polynomial $P$ on $T$ is a degree $(p,p,p)$ polynomial. Furthermore, $P$ vanishes on all tensors of rank at most $r$ if and only if $P(T)=0$ as a polynomial. To check the last condition we may build a matrix $N_{d,r,p}$ with:
\begin{enumerate}
\item $\binom{p+d^3-1}{p}$ columns indexed by $\sC^{[d]^3}_p$, equivalently monomials of degree $p$,
\item $\binom{p+dr-1}{p}^3$ rows indexed by monomials of degree $(p,p,p)$ in variables $(v_{i,j})_k$,
\item the entry in a column indexed by a monomial $g\in \sC^{[d]^3}_p$ and a row indexed by a monomial $m$ is the coefficient of $m$ in the polynomial $g(T)$.
\end{enumerate}
The kernel of $N_{d,r,p}$ is naturally identified with the space of homogeneous degree $p$ polyomials vanishing on rank $r$ tensors. In particular, the kernel is trivial if and only if there are no such polynomials. While this approach is very general, it is also not applicable in most cases due to the large sizes of the matrices $N_{d,r,p}$.

A better approach to understand the equations of secant varieties is to exploit group actions. For this, one decomposes the space $S^p(\C^d\otimes\C^d\otimes \C^d)$ of homogeneous degree $p$ polynomials on the tensor space as a direct sum of irreducible $G:=GL(n)\times GL(n)\times GL(n)$ representations. The irreducible polynomial representations of $G$ are indexed by triples of Young diagrams, in our case each one in the triple will have $p$ elements:
\[
S^p(\C^d\otimes\C^d\otimes \C^d)=\bigoplus V_{\lambda,\mu,\nu}^{\bigoplus g_{\lambda,\mu,\nu}}\,,
\]
where $V_{\lambda,\mu,\nu}$ is the triple corresponding to three Young diagrams $\lambda,\mu,\nu$ with $p$ boxes each and $ V_{\lambda,\mu,\nu}
$ is the Kronecker coefficient. The vector space of homogeneous polynomials of degree $p$ that vanish on the $n$\textsuperscript{th} secant variety is a subrepresentation and hence also may be decomposed as:
\[
\bigoplus V_{\lambda,\mu,\nu}^{\bigoplus a_{\lambda,\mu,\nu}}\,,
\]
where now the coefficients $a_{\lambda,\mu,\nu}\leq g_{\lambda,\mu,\nu}$ are unknown. One way to determine them is to consider highest weigh space $HWS_{\lambda,\mu,\nu}$, that is a vector space of dimension $g_{\lambda,\mu,\nu}$ inside $V_{\lambda,\mu,\nu}^{\bigoplus g_{\lambda,\mu,\nu}}.$
It turns out that $a_{\lambda,\mu,\nu}$ equals the dimension of the subspace of $HWS_{\lambda,\mu,\nu}$ consisting on those polynomials that vanish on the $n$\textsuperscript{th} secant variety. Finding the basis of $HWS_{\lambda,\mu,\nu}$ and efficiently evaluating it on points of the secant variety is an art on its own. However, in many examples, the procedure above was successfully carried out in practice. We refer to \cite{breiding2022equations, hauenstein2013equations} for details.

\begin{Exa}
The generic rank in $\C^7\ot\C^7\ot\C^7$ is $19$. Tensors of border rank $18$ form a hypersurface with the defining equation of degree at least $187000$ \cite[Section 3.2]{hauenstein2013equations}. Applying Corollary \ref{cor:noeqthenlowrank} we see that every tensor in that space has asymptotic rank smaller than $18.25$. 
\end{Exa}

We believe that many further results on asymptotic rank may be obtained in a similar way, through Corollary \ref{cor:noeqthenlowrank} and \ref{cor:noeqthenlowrankimplicit}, especially in combination with Remark \ref{rem:improvements}. We leave this for future work.

\section*{Acknowledgments}

We would like to thank Andreas Bj\"orklund and Joseph Landsberg
for useful discussions and important comments. The second author is supported
by the DFG grant 467575307.




\begin{thebibliography}{10}

\bibitem{alexeev2011tensor}
B.~Alexeev, M.~A. Forbes, and J.~Tsimerman.
\newblock Tensor rank: Some lower and upper bounds.
\newblock In {\em 2011 IEEE 26th Annual Conference on Computational
  Complexity}, pages 283--291. IEEE, 2011.

\bibitem{AlmanW2021}
J.~Alman and V.~Vassilevska~Williams.
\newblock A refined laser method and faster matrix multiplication.
\newblock In {\em Proceedings of the 2021 {ACM}-{SIAM} {S}ymposium on
  {D}iscrete {A}lgorithms ({SODA})}, pages 522--539. [Society for Industrial
  and Applied Mathematics (SIAM)], Philadelphia, PA, 2021.

\bibitem{AlmanZ2023}
J.~Alman and H.~Zhang.
\newblock Generalizations of matrix multiplication can solve the light bulb
  problem.
\newblock In {\em 64th {IEEE} Annual Symposium on Foundations of Computer
  Science, {FOCS} 2023, Santa Cruz, CA, USA, November 6-9, 2023}, pages
  1471--1495. {IEEE}, 2023.

\bibitem{bernardi2013cactus}
A.~Bernardi and K.~Ranestad.
\newblock On the cactus rank of cubic forms.
\newblock {\em Journal of Symbolic Computation}, 50:291--297, 2013.

\bibitem{bernardi2020waring}
A.~Bernardi and D.~Taufer.
\newblock Waring, tangential and cactus decompositions.
\newblock {\em Journal de Math{\'e}matiques Pures et Appliqu{\'e}es},
  143:1--30, 2020.

\bibitem{BiniCRL1979}
D.~Bini, M.~Capovani, F.~Romani, and G.~Lotti.
\newblock {$O(n\sp{2.7799})$} complexity for {$n\times n$} approximate matrix
  multiplication.
\newblock {\em Inform. Process. Lett.}, 8(5):234--235, 1979.

\bibitem{BjorklundK2023}
A.~Bj{\"{o}}rklund and P.~Kaski.
\newblock The asymptotic rank conjecture and the set cover conjecture are not
  both true.
\newblock {\em CoRR}, abs/2310.11926, 2023.

\bibitem{Blaser1999}
M.~Bl\"{a}ser.
\newblock A {$\frac52 n^2$}-lower bound for the rank of {$n\times n$}-matrix
  multiplication over arbitrary fields.
\newblock In {\em 40th {A}nnual {S}ymposium on {F}oundations of {C}omputer
  {S}cience ({N}ew {Y}ork, 1999)}, pages 45--50. IEEE Computer Soc., Los
  Alamitos, CA, 1999.

\bibitem{BlaserCZ2018}
M.~Bl{\"{a}}ser, M.~Christandl, and J.~Zuiddam.
\newblock The border support rank of two-by-two matrix multiplication is seven.
\newblock {\em Chic. J. Theor. Comput. Sci.}, 2018.

\bibitem{BlasiakCGPU2023}
J.~Blasiak, H.~Cohn, J.~A. Grochow, K.~Pratt, and C.~Umans.
\newblock Matrix multiplication via matrix groups.
\newblock In Y.~T. Kalai, editor, {\em 14th Innovations in Theoretical Computer
  Science Conference, {ITCS} 2023, January 10-13, 2023, MIT, Cambridge,
  Massachusetts, {USA}}, volume 251 of {\em LIPIcs}, pages 19:1--19:16. Schloss
  Dagstuhl - Leibniz-Zentrum f{\"{u}}r Informatik, 2023.

\bibitem{breiding2022equations}
P.~Breiding, R.~Hodges, C.~Ikenmeyer, and M.~Micha{\l}ek.
\newblock Equations for {GL} invariant families of polynomials.
\newblock {\em Vietnam Journal of Mathematics}, 50(2):545--556, 2022.

\bibitem{buczynska2014secant}
W.~Buczy{\'n}ska and J.~Buczy{\'n}ski.
\newblock Secant varieties to high degree veronese reembeddings, catalecticant
  matrices and smoothable gorenstein schemes.
\newblock {\em Journal of Algebraic Geometry}, 23(1):63--90, 2014.

\bibitem{buczynska2021apolarity}
W.~Buczy{\'n}ska and J.~Buczy{\'n}ski.
\newblock Apolarity, border rank, and multigraded hilbert scheme.
\newblock {\em Duke Mathematical Journal}, 170(16):3659--3702, 2021.

\bibitem{buczynski2013ranks}
J.~Buczy{\'n}ski and J.~M. Landsberg.
\newblock Ranks of tensors and a generalization of secant varieties.
\newblock {\em Linear Algebra and its Applications}, 438(2):668--689, 2013.

\bibitem{Burgisser1990}
P.~B\"urgisser.
\newblock {\em Degenerationsordnung und Tr\"agerfunktional bilinearer
  Abbildungen}.
\newblock PhD thesis, Universit\"at Konstanz, 1990.

\bibitem{burgisser2013algebraic}
P.~B{\"u}rgisser, M.~Clausen, and M.~A. Shokrollahi.
\newblock {\em Algebraic Complexity Theory}.
\newblock Springer Science \& Business Media, 2013.

\bibitem{ChristandlVZ2021}
M.~Christandl, P.~Vrana, and J.~Zuiddam.
\newblock Barriers for fast matrix multiplication from irreversibility.
\newblock {\em Theory Comput.}, 17:Paper No. 2, 32, 2021.

\bibitem{ChristandlVZ2023}
M.~Christandl, P.~Vrana, and J.~Zuiddam.
\newblock Universal points in the asymptotic spectrum of tensors.
\newblock {\em J. Amer. Math. Soc.}, 36(1):31--79, 2023.

\bibitem{CohnKSU2005}
H.~Cohn, R.~D. Kleinberg, B.~Szegedy, and C.~Umans.
\newblock Group-theoretic algorithms for matrix multiplication.
\newblock In {\em 46th Annual {IEEE} Symposium on Foundations of Computer
  Science {(FOCS} 2005), 23-25 October 2005, Pittsburgh, PA, USA, Proceedings},
  pages 379--388. {IEEE} Computer Society, 2005.

\bibitem{CohnU2003}
H.~Cohn and C.~Umans.
\newblock A group-theoretic approach to fast matrix multiplication.
\newblock In {\em 44th Symposium on Foundations of Computer Science {(FOCS}
  2003), 11-14 October 2003, Cambridge, MA, USA, Proceedings}, pages 438--449.
  {IEEE} Computer Society, 2003.

\bibitem{CohnU2012}
H.~Cohn and C.~Umans.
\newblock Fast matrix multiplication using coherent configurations.
\newblock In {\em Proceedings of the {T}wenty-{F}ourth {A}nnual {ACM}-{SIAM}
  {S}ymposium on {D}iscrete {A}lgorithms}, pages 1074--1087. SIAM,
  Philadelphia, PA, 2012.

\bibitem{ConnerGLV2022}
A.~Conner, F.~Gesmundo, J.~M. Landsberg, and E.~Ventura.
\newblock Rank and border rank of {K}ronecker powers of tensors and
  {S}trassen's laser method.
\newblock {\em Comput. Complexity}, 31(1):Paper No. 1, 40, 2022.

\bibitem{ConnerGLVW2021}
A.~Conner, F.~Gesmundo, J.~M. Landsberg, E.~Ventura, and Y.~Wang.
\newblock Towards a geometric approach to {S}trassen's asymptotic rank
  conjecture.
\newblock {\em Collect. Math.}, 72(1):63--86, 2021.

\bibitem{Conner_Harper_Landsberg_2023}
A.~Conner, A.~Harper, and J.~Landsberg.
\newblock New lower bounds for matrix multiplication and $\operatorname
  {det}_3$.
\newblock {\em Forum of Mathematics, Pi}, 11:e17, 2023.

\bibitem{CoppersmithW1982}
D.~Coppersmith and S.~Winograd.
\newblock On the asymptotic complexity of matrix multiplication.
\newblock {\em SIAM J. Comput.}, 11(3):472--492, 1982.

\bibitem{CoppersmithW1990}
D.~Coppersmith and S.~Winograd.
\newblock Matrix multiplication via arithmetic progressions.
\newblock {\em J. Symbolic Comput.}, 9(3):251--280, 1990.

\bibitem{CoxLO2015}
D.~A. Cox, J.~Little, and D.~O'Shea.
\newblock {\em Ideals, Varieties, and Algorithms}.
\newblock Undergraduate Texts in Mathematics. Springer, Cham, fourth edition,
  2015.

\bibitem{CyganDLMNOPSW2016}
M.~Cygan, H.~Dell, D.~Lokshtanov, D.~Marx, J.~Nederlof, Y.~Okamoto, R.~Paturi,
  S.~Saurabh, and M.~Wahlstr{\"{o}}m.
\newblock On problems as hard as {CNF-SAT}.
\newblock {\em {ACM} Trans. Algorithms}, 12(3):41:1--41:24, 2016.

\bibitem{CyganFKLMPPS15}
M.~Cygan, F.~V. Fomin, L.~Kowalik, D.~Lokshtanov, D.~Marx, M.~Pilipczuk,
  M.~Pilipczuk, and S.~Saurabh.
\newblock {\em Parameterized Algorithms}.
\newblock Springer, 2015.

\bibitem{DuanWZ2023}
R.~Duan, H.~Wu, and R.~Zhou.
\newblock Faster matrix multiplication via asymmetric hashing.
\newblock In {\em 64th {IEEE} Annual Symposium on Foundations of Computer
  Science, {FOCS} 2023, Santa Cruz, CA, USA, November 6-9, 2023}, pages
  2129--2138. {IEEE}, 2023.

\bibitem{galkazka2017vector}
M.~Ga{\l}{\k{a}}zka.
\newblock Vector bundles give equations of cactus varieties.
\newblock {\em Linear Algebra and its Applications}, 521:254--262, 2017.

\bibitem{galkazka2023multigraded}
M.~Ga{\l}{\k{a}}zka.
\newblock Multigraded apolarity.
\newblock {\em Mathematische Nachrichten}, 296(1):286--313, 2023.

\bibitem{Gartenberg1985}
P.~A. Gartenberg.
\newblock {\em Fast Rectangular Matrix Multiplication}.
\newblock PhD thesis, University of California, Los Angeles, 1985.

\bibitem{Hastad1990}
J.~H{\aa}stad.
\newblock Tensor rank is {NP}-complete.
\newblock {\em J. Algorithms}, 11(4):644--654, 1990.

\bibitem{hauenstein2013equations}
J.~D. Hauenstein, C.~Ikenmeyer, and J.~M. Landsberg.
\newblock Equations for lower bounds on border rank.
\newblock {\em Experimental Mathematics}, 22(4):372--383, 2013.

\bibitem{HillarL2013}
C.~J. Hillar and L.-H. Lim.
\newblock Most tensor problems are {NP}-hard.
\newblock {\em J. ACM}, 60(6):Art. 45, 39, 2013.

\bibitem{iarrobino1999power}
A.~Iarrobino and V.~Kanev.
\newblock {\em Power sums, Gorenstein algebras, and determinantal loci}.
\newblock Springer Science \& Business Media, 1999.

\bibitem{KarppaK2019}
M.~Karppa and P.~Kaski.
\newblock Probabilistic tensors and opportunistic {B}oolean matrix
  multiplication.
\newblock In {\em Proceedings of the {T}hirtieth {A}nnual {ACM}-{SIAM}
  {S}ymposium on {D}iscrete {A}lgorithms}, pages 496--515. SIAM, Philadelphia,
  PA, 2019.

\bibitem{KrauthgamerT2019}
R.~Krauthgamer and O.~Trabelsi.
\newblock The set cover conjecture and subgraph isomorphism with a tree
  pattern.
\newblock In R.~Niedermeier and C.~Paul, editors, {\em 36th International
  Symposium on Theoretical Aspects of Computer Science, {STACS} 2019, March
  13-16, 2019, Berlin, Germany}, volume 126 of {\em LIPIcs}, pages 45:1--45:15.
  Schloss Dagstuhl - Leibniz-Zentrum f{\"{u}}r Informatik, 2019.

\bibitem{landsberg2015nontriviality}
J.~Landsberg.
\newblock Nontriviality of equations and explicit tensors in
  $\mathbb{C}^m\otimes \mathbb{C}^m\otimes\mathbb{C}^m$ of border rank at least
  $2m-2$.
\newblock {\em Journal of Pure and Applied Algebra}, 219(8):3677--3684, 2015.

\bibitem{landsberg2019towards}
J.~Landsberg and M.~Micha{\l}ek.
\newblock Towards finding hay in a haystack: explicit tensors of border rank
  greater than $2.02m$ in $ \mathbb{C}^m\otimes
  \mathbb{C}^m\otimes\mathbb{C}^m$.
\newblock {\em arXiv preprint arXiv:1912.11927}, 2019.

\bibitem{Landsberg2006}
J.~M. Landsberg.
\newblock The border rank of the multiplication of {$2\times2$} matrices is
  seven.
\newblock {\em J. Amer. Math. Soc.}, 19(2):447--459, 2006.

\bibitem{Landsberg2012}
J.~M. Landsberg.
\newblock {\em Tensors: Geometry and Applications}, volume 128 of {\em Graduate
  Studies in Mathematics}.
\newblock American Mathematical Society, Providence, RI, 2012.

\bibitem{landsberg2014new}
J.~M. Landsberg.
\newblock New lower bounds for the rank of matrix multiplication.
\newblock {\em SIAM Journal on Computing}, 43(1):144--149, 2014.

\bibitem{Landsberg2019}
J.~M. Landsberg.
\newblock {\em Tensors: Asymptotic Geometry and Developments 2016--2018},
  volume 132 of {\em CBMS Regional Conference Series in Mathematics}.
\newblock American Mathematical Society, Providence, RI, 2019.

\bibitem{landsberg2004ideals}
J.~M. Landsberg and L.~Manivel.
\newblock On the ideals of secant varieties of {S}egre varieties.
\newblock {\em Foundations of Computational Mathematics}, 4:397--422, 2004.

\bibitem{landsberg2017geometry}
J.~M. Landsberg and M.~Micha{\l}ek.
\newblock On the geometry of border rank decompositions for matrix
  multiplication and other tensors with symmetry.
\newblock {\em SIAM Journal on Applied Algebra and Geometry}, 1(1):2--19, 2017.

\bibitem{landsberg2018lower}
J.~M. Landsberg and M.~Micha{\l}ek.
\newblock A lower bound for the border rank of matrix multiplication.
\newblock {\em International Mathematics Research Notices},
  2018(15):4722--4733, 2018.

\bibitem{landsberg2013equations}
J.~M. Landsberg and G.~Ottaviani.
\newblock Equations for secant varieties of veronese and other varieties.
\newblock {\em Annali di Matematica Pura ed Applicata}, 192(4):569--606, 2013.

\bibitem{landsberg2015new}
J.~M. Landsberg and G.~Ottaviani.
\newblock New lower bounds for the border rank of matrix multiplication.
\newblock {\em Theory of Computing}, 11(1):285--298, 2015.

\bibitem{landsberg2010ranks}
J.~M. Landsberg and Z.~Teitler.
\newblock On the ranks and border ranks of symmetric tensors.
\newblock {\em Foundations of Computational Mathematics}, 10(3):339--366, 2010.

\bibitem{LeGall2014}
F.~Le~Gall.
\newblock Powers of tensors and fast matrix multiplication.
\newblock In {\em I{SSAC} 2014---{P}roceedings of the 39th {I}nternational
  {S}ymposium on {S}ymbolic and {A}lgebraic {C}omputation}, pages 296--303.
  ACM, New York, 2014.

\bibitem{Mauch1998}
F.~Mauch.
\newblock {\em Ein Randverteilungsproblem und seine Anwendung auf das
  asymptotische Spektrum bilinearer Abbildungen}.
\newblock PhD thesis, Universit\"at Konstanz, 1998.

\bibitem{michalek2021invitation}
M.~Micha{\l}ek and B.~Sturmfels.
\newblock {\em Invitation to Nonlinear Algebra}, volume 211.
\newblock American Mathematical Society, 2021.

\bibitem{Pan1978}
V.~Y. Pan.
\newblock Strassen's algorithm is not optimal. {T}rilinear technique of
  aggregating, uniting and canceling for constructing fast algorithms for
  matrix operations.
\newblock In {\em 19th {A}nnual {S}ymposium on {F}oundations of {C}omputer
  {S}cience ({A}nn {A}rbor, {M}ich., 1978)}, pages 166--176. IEEE, Long Beach,
  CA, 1978.

\bibitem{Pratt2023}
K.~Pratt.
\newblock A stronger connection between the asymptotic rank conjecture and the
  set cover conjecture.
\newblock {\em CoRR}, abs/2311.02774, 2023.

\bibitem{Raz2013}
R.~Raz.
\newblock Tensor-rank and lower bounds for arithmetic formulas.
\newblock {\em J. ACM}, 60(6):Art. 40, 15, 2013.

\bibitem{Romani1982}
F.~Romani.
\newblock Some properties of disjoint sums of tensors related to matrix
  multiplication.
\newblock {\em SIAM J. Comput.}, 11(2):263--267, 1982.

\bibitem{Schonhage1981}
A.~Sch\"{o}nhage.
\newblock Partial and total matrix multiplication.
\newblock {\em SIAM J. Comput.}, 10(3):434--455, 1981.

\bibitem{ZhaoWZ2019}
Z.~Song, D.~P. Woodruff, and P.~Zhong.
\newblock Relative error tensor low rank approximation.
\newblock In {\em Proceedings of the {T}hirtieth {A}nnual {ACM}-{SIAM}
  {S}ymposium on {D}iscrete {A}lgorithms}, pages 2772--2789. SIAM,
  Philadelphia, PA, 2019.

\bibitem{Stothers2010}
A.~J. Stothers.
\newblock {\em On the Complexity of Matrix Multiplication}.
\newblock PhD thesis, University of Edinburgh, 2010.

\bibitem{Strassen1969}
V.~Strassen.
\newblock Gaussian elimination is not optimal.
\newblock {\em Numer. Math.}, 13:354--356, 1969.

\bibitem{Strassen1986}
V.~Strassen.
\newblock The asymptotic spectrum of tensors and the exponent of matrix
  multiplication.
\newblock In {\em 27th Annual Symposium on Foundations of Computer Science,
  Toronto, Canada, 27-29 October 1986}, pages 49--54. {IEEE} Computer Society,
  1986.

\bibitem{Strassen1987}
V.~Strassen.
\newblock Relative bilinear complexity and matrix multiplication.
\newblock {\em J. Reine Angew. Math.}, 375/376:406--443, 1987.

\bibitem{Strassen1988}
V.~Strassen.
\newblock The asymptotic spectrum of tensors.
\newblock {\em J. Reine Angew. Math.}, 384:102--152, 1988.

\bibitem{Strassen1991}
V.~Strassen.
\newblock Degeneration and complexity of bilinear maps: some asymptotic
  spectra.
\newblock {\em J. Reine Angew. Math.}, 413:127--180, 1991.

\bibitem{Strassen1994}
V.~Strassen.
\newblock Algebra and complexity.
\newblock In {\em First {E}uropean {C}ongress of {M}athematics, {V}ol. {II}
  ({P}aris, 1992)}, volume 120 of {\em Progr. Math.}, pages 429--446.
  Birkh\"{a}user, Basel, 1994.

\bibitem{Strassen2005}
V.~Strassen.
\newblock Komplexit\"{a}t und {G}eometrie bilinearer {A}bbildungen.
\newblock {\em Jahresber. Deutsch. Math.-Verein.}, 107(1):3--31, 2005.

\bibitem{teitler2014geometric}
Z.~Teitler.
\newblock Geometric lower bounds for generalized ranks.
\newblock {\em arXiv preprint arXiv:1406.5145}, 2014.

\bibitem{Tobler1991}
V.~Tobler.
\newblock {\em Spezialisierung und Degeneration von Tensoren}.
\newblock PhD thesis, Universit\"at Konstanz, 1991.

\bibitem{VassilevskaWilliams2012}
V.~{Vassilevska Williams}.
\newblock Multiplying matrices faster than {C}oppersmith-{W}inograd [extended
  abstract].
\newblock In {\em S{TOC}'12---{P}roceedings of the 2012 {ACM} {S}ymposium on
  {T}heory of {C}omputing}, pages 887--898. ACM, New York, 2012.

\bibitem{WigdersonZ2023}
A.~Wigderson and J.~Zuiddam.
\newblock Asymptotic spectra: {T}heory, applications and extensions.
\newblock Manuscript dated October 24, 2023; available at
  \url{https://www.math.ias.edu/~avi/PUBLICATIONS/WigdersonZu_Final_Draft_Oct2023.pdf},
  2023.

\bibitem{zak1993tangents}
F.~L. Zak.
\newblock {\em Tangents and Secants of Algebraic Varieties}, volume 127.
\newblock American Mathematical Society, 1993.

\end{thebibliography}
\end{document}